\documentclass[a4paper,10pt,twocolumn]{article}
\usepackage[footnotesize]{caption}
\usepackage[caption=false,font=footnotesize]{subfig}
\usepackage{graphicx,amsmath,amsfonts,amsthm,amssymb,url,cite,bm,bbm}
\usepackage[margin={0.575in,0.8in}]{geometry}
\usepackage{tikz,booktabs,pifont}

\newcommand{\cmark}{\ding{51}}%
\newcommand{\xmark}{\ding{55}}%

\begin{document}

\newcommand{\eps}{\varepsilon}
\newcommand{\var}{\operatorname{Var}}
\newcommand{\pr}{\mathbb{P}}
\newcommand{\expn}{\mathbb{E}}
\newcommand{\caps}{C^{\text{s}}}
\newcommand{\capj}{C^{\text{j}}}
\newcommand{\thint}{\vec{\theta}_{\text{int}}}
\newcommand{\thall}{\vec{\theta}_{\text{all1}}}
\newcommand{\thmaj}{\vec{\theta}_{\text{maj}}}
\newcommand{\thmin}{\vec{\theta}_{\text{min}}}
\newcommand{\thcoi}{\vec{\theta}_{\text{coin}}}
\newcommand{\pint}{p_{\text{int}}}
\newcommand{\pall}{p_{\text{all1}}}
\newcommand{\pmaj}{p_{\text{maj}}}
\newcommand{\pmin}{p_{\text{min}}}
\newcommand{\pcoi}{p_{\text{coin}}}
\renewcommand{\vec}{\bm}
\newcommand{\vc}{\mathbf}
\newcommand{\kl}{d_{\operatorname{KL}}}

\newtheorem{theorem}{Theorem}
\newtheorem{lemma}{Lemma}
\newtheorem{proposition}{Proposition}
\newtheorem{corollary}{Corollary}
\newtheorem{conjecture}{Conjecture}

\title{Optimal sequential fingerprinting: Wald vs. Tardos}
\author{Thijs Laarhoven\footnote{T. Laarhoven is with the Department of Mathematics and Computer Science, Eindhoven University of Technology, P.O. Box 513, 5600 MB Eindhoven, The Netherlands. \protect\\
E-mail: mail@thijs.com.}}
\date{\today}

\maketitle
\begin{abstract}
We study sequential collusion-resistant fingerprinting, where the fingerprinting code is generated in advance but accusations may be made between rounds, and show that in this setting both the dynamic Tardos scheme and schemes building upon Wald's sequential probability ratio test (SPRT) are asymptotically optimal. We further compare these two approaches to sequential fingerprinting, highlighting differences between the two schemes. Based on these differences, we argue that Wald's scheme should in general be preferred over the dynamic Tardos scheme, even though both schemes have their merits. As a side result, we derive an optimal sequential group testing method for the classical model, which can easily be generalized to different group testing models.
\end{abstract}


\section{Introduction}

In collusion-resistant fingerprinting, a distributor aims to embed fingerprints in digital content so that even if several users collude and mix their fingerprinted copies into a new copy, the resulting pirate version can still be traced back to the guilty parties. In 2003, the seminal work of Tardos~\cite{tardos08} showed that in the non-adaptive setting, fingerprinting codes with this property must have a length $\ell$ quadratic in the number of colluders $c$ and logarithmic in the total number of users $n$ (i.e., $\ell \propto c^2 \log n$), and that such codes exist. These codes guarantee that with a proper decoding algorithm, at least one of the colluders can be found with high probability. Later, in 2013 it was shown~\cite{laarhoven13tit} that in the adaptive setting, where code words are sent out symbol by symbol and the distributor is allowed to base future decisions on previous results, in fact \textit{all} colluders can provably be found with a code length $\ell \propto c^2 \log n$, using a dynamic version of Tardos' scheme. Results in fingerprinting have recently found applications in other fields as well, including group testing~\cite{laarhoven13allerton, meerwald11b} and differential privacy~\cite{bun14, dwork14, steinke14, ullman13}.

Over the years, various follow-up works to Tardos' milestone paper have allowed us to understand why Tardos' scheme is designed the way it is designed~\cite{furon08, skoric13}, how the scheme can be further improved theoretically~\cite{blayer08, laarhoven14dcc, nuida07, nuida09, skoric08b, skoric08, skoric13} and practically~\cite{charpentier09, desoubeaux13, furon09, furon09b, furon09c, furon12, furon14, kuribayashi12, meerwald11, meerwald12, perez09, simone11}, what are the limitations of fingerprinting in general~\cite{amiri09, huang12, moulin08} and of the optimized (symmetric) Tardos scheme~\cite{laarhoven13ihmmsec, laarhoven14dcc}, and how these limitations can be overcome by further modifying the scheme~\cite{ibrahimi14, laarhoven14ihmmsec, oosterwijk13} to achieve asymptotic optimality~\cite{amiri09, huang12, moulin08, oosterwijk13b}. Most notably, connections were made between fingerprinting, game theory, channel coding, and statistical hypothesis testing, which ultimately allowed us to explain why the optimal non-adaptive designs are optimal~\cite{abbe12, huang12, laarhoven14ihmmsec, meerwald12}.

Although various of these insights directly carry over to the adaptive setting, in this area several questions remain:
\begin{itemize}\vspace{-0.1cm}
	\item Is the ``dynamic Tardos scheme''~\cite{laarhoven13tit, laarhoven13wifs} optimal? \vspace{-0.2cm}
	\item What motivates the design of this scheme?\vspace{-0.1cm}
\end{itemize}
Answering these and related questions may ultimately lead to the same level of understanding for the adaptive case as for the non-adaptive setting, allowing practitioners to make well-motivated design choices in the adaptive setting as well. \\

\textbf{Contributions.} In this paper we answer the second question by showing a connection with what is known in the literature as the \textit{sequential probability ratio test (SPRT)}, invented by Wald in the 1940s~\cite{wald47}. As a result, we are also able to take a first step towards answering the first question: within the class of \textit{sequential} fingerprinting schemes, where the code book is not constructed adaptively, both the dynamic Tardos scheme and schemes built from Wald's SPRT are essentially optimal for the uninformed fingerprinting game. We discuss in detail how sequential fingerprinting schemes can naturally be constructed from Wald's SPRT, and how various results from the literature can be used to tune these schemes to different scenarios. We finally compare the dynamic Tardos scheme to Wald's SPRT, and highlight why in general Wald's scheme should be preferred. \\

\textbf{Roadmap.} First, in Section~\ref{sec:model} we outline the fingerprinting model considered in this paper. In Section~\ref{sec:tardos} we briefly review the dynamic Tardos scheme and its variants. Then, in Section~\ref{sec:wald} we describe Wald's sequential probability ratio test procedure, and how it can be applied to fingerprinting to obtain optimal sequential fingerprinting schemes. Next, in Section~\ref{sec:comp} we illustrate the similarities and differences between these schemes through explicit examples, on the way showing that both schemes are asymptotically optimal for the uninformed fingerprinting game. Finally, in Section~\ref{sec:overview} we give an overview of the main characteristics of both schemes, which may allow practitioners to make a well-informed choice between the two schemes.

\section{Model}
\label{sec:model}

The collusion-resistant fingerprinting problem is often modeled as the following two-person game between the distributor $\mathcal{D}$ and the coalition of pirates $\mathcal{C}$. The set of colluders is assumed to be a random subset of size $|\mathcal{C}| = c$ from the complete set of $n$ users $\mathcal{U}$, and the identities of these colluders are unknown to the distributor. The distributor might not know $c$ either, and he may only have a (crude) upper bound $c_0 \geq c$ on $c$. The aim of the game for the distributor is to discover the identities of the colluders without accidentally accusing innocent users, with as little effort as possible. The colluders want to prevent this and remain hidden. The game consists of the following three phases: the distributor uses an \textit{encoder} to generate fingerprints; the colluders employ a \textit{collusion channel} to generate pirate output; and the distributor uses a \textit{decoder} to map pirate output to a set of accused users. We describe these three phases below.

\paragraph{a. Encoder}
The distributor generates a code $\mathcal{X} = \{\vc{x}_1, \dots, \vc{x}_n\}$, consisting of binary code words $\vc{x}_j \in \{0,1\}^{\ell}$ for each user $j \in \mathcal{U}$, and each column $i \in \{1, \dots, \ell\}$ corresponds to a different segment of the content.\footnote{More generally $\mathcal{X}$ is a code with entries from an alphabet of size $q$, but here we restrict our attention to the case $q = 2$.} A common restriction on the encoding process is to assume that $\mathcal{X}$ is created by first generating a \textit{bias vector} $\vc{p} \in [0,1]^{\ell}$, by choosing each entry $p_i$ independently from a certain distribution $F$, and then generating $\mathcal{X}$ using $P((\vc{x}_j)_i = 1) = p_i$. Schemes with this property are sometimes called \textit{bias-based schemes}.

As initially suggested by Tardos~\cite{tardos08} and later proven by Huang and Moulin~\cite{huang12}, the best way to build bias-based encoders (for large coalitions, in the uninformed setting) is to use the arcsine distribution for generating $p_i$'s. For this distribution, we have the following distribution function $F$:
\begin{align}
F(p) = \frac{2}{\pi} \arcsin \sqrt{p} . \qquad (p \in (0,1)) \label{dist1}
\end{align}
Unless stated otherwise, throughout the paper we will assume that this encoder is used for generating biases.

\paragraph{b. Collusion channel}
Given $\mathcal{X}$, the entries can be used to select and embed watermarks in the content, and the content is sent out to the users. The colluders get together, compare their copies, and use a \textit{collusion channel} or pirate strategy $\Theta$ to select the pirate output $\vc{y} \in \{0,1\}^{\ell}$. If the pirate attack is symmetrical both in the colluders and in the positions $i$, then the collusion channel can be modeled by a vector $\boldsymbol{\theta} \in [0,1]^{c+1}$, with entries $\theta_z = P(Y_i = 1 \mid Z = z)$ indicating the probability of outputting a $1$ when pirates receive $z$ ones and $c - z$ zeros.

\paragraph{c. Decoder}
After the pirate output has been distributed, the distributor intercepts it and applies a decoding algorithm to $\mathcal{X}, \vc{y}, \vc{p}$ to compute a set $\mathcal{C}' \subseteq \mathcal{U}$ of accused users. This is commonly done by assigning scores to users, and accusing those users whose scores exceed a predefined threshold $\eta$. The distributor wins the game if $\mathcal{C}' = \mathcal{C}$ and loses\footnote{In this paper we consider the catch-all scenario, where not \textit{at least one} colluder (the catch-one scenario) but \textit{all} colluders should be found for the distributor to win the game.} if $\mathcal{C}' \neq \mathcal{C}$, which could be because an innocent user $j \notin \mathcal{C}$ is falsely accused (a false positive error), or because a guilty user $j \in \mathcal{C}$ is not accused (a false negative error). We often write $\eps_1$ and $\eps_2$ for (upper bounds on) the false positive and false negative probabilities. \\

Finally, the differences between non-adaptive (static) fingerprinting, adaptive (dynamic) fingerprinting, and sequential fingerprinting can be explained by showing in which order these phases take place. Denoting by $a_i, b_i, c_i$ the three phases corresponding to the $i$th segment of the content, we can order the phases as follows:
\begin{itemize}\vspace{-0.1cm}
	\item Non-adaptive: $a_{[1, \dots, \ell]}; \, b_{[1, \dots, \ell]}; \, c_{[1, \dots, \ell]}$.\vspace{-0.1cm}
	\item Sequential: $a_{[1, \dots, \ell]}; \, b_1; \, c_1; \, b_2; \, c_2; \, \dots; \, b_{\ell}; \, c_{\ell}$.\vspace{-0.1cm}
	\item Adaptive: $a_1; \, b_1; \, c_1; \, a_2; \, b_2; \, c_2; \, \dots; \, a_{\ell}; \, b_{\ell}; \, c_{\ell}$.\vspace{-0.1cm}
\end{itemize}
In other words: in the adaptive setting the code can be adjusted and accusations can be made after every symbol; in sequential fingerprinting only users can be accused between rounds; and in non-adaptive settings the distributor is only allowed to make a final decision at the end of the game.

While most work in the literature focuses on the non-adaptive setting, some work has also been done on sequential~\cite{safavi03} and adaptive fingerprinting~\cite{berkman01, fiat01, laarhoven13tit, laarhoven13wifs, roelse11, tassa05}. In this paper we will mostly deal with the sequential setting.

\section{Tardos' scheme}
\label{sec:tardos}

\subsection{Non-adaptive scheme}

In Tardos' original scheme~\cite{tardos08} and many of its subsequent variants, decoding in the non-adaptive setting is done as follows. First, for each segment $i$ and user $j$, scores $S_{j,i} = g(x_{j,i}, y_i, p_i)$ are assigned using a score function $g$. Then, in the non-adaptive setting, a user $j \in \mathcal{U}$ is accused iff $S_j = \sum_{i = 1}^{\ell} S_{j,i} > \eta$ for some well-chosen threshold $\eta$. Choosing a suitable score function is crucial, and it was long thought that the following symmetrized version~\cite{skoric08} of Tardos' original proposal was the best choice:
\begin{align}
g(x,y,p) = \begin{cases} 
\sqrt{p/(1 - p)} & \text{if } (x, y) = (0,0); \\
-\sqrt{p/(1 - p)} & \text{if } (x, y) = (0,1); \\
-\sqrt{(1 - p)/p} & \text{if } (x, y) = (1,0); \\
\sqrt{(1 - p)/p} & \text{if } (x, y) = (1,1).
\end{cases}
\end{align}
This function turns out to work quite well against arbitrary pirate attacks, and it has the convenient property that regardless of the pirate strategy, one always has $\expn(S_{j,i} \mid H_0) = 0$, $\expn(S_{j,i}^2 \mid H_0) = 1$, and $\expn(S_{j,i} \mid H_1) \approx \frac{2}{\pi}$, where the hypotheses $H_0$ and $H_1$ correspond to:
\begin{itemize}
	\item $H_0$: user $j$ is innocent ($j \notin \mathcal{C}$).
	\item $H_1$: user $j$ is guilty ($j \in \mathcal{C}$).
\end{itemize}
As convenient as this decoder may be, it is known to be suboptimal~\cite{huang12, laarhoven13ihmmsec}, with code lengths which are up to a factor $\frac{1}{4} \pi^2 \approx 2.47$ longer than required. Using various different approaches (e.g.\ Lagrange optimization~\cite{oosterwijk13, oosterwijk13b}, Neyman-Pearson decoding~\cite{laarhoven14ihmmsec}, Bayesian decoding~\cite{desoubeaux13}, MAP decoding~\cite{furon09c}, empirical mutual information decoding~\cite{moulin08}) it was later found that there are various ways to construct decoders for the uninformed setting in fingerprinting with a better performance than the symmetric score function. Various of these decoders were recently benchmarked in~\cite{furon14}, indicating that different decoders work better in different settings. For comparison with Wald's SPRT we will continue the description of Tardos' scheme using Neyman-Pearson-motivated decoders, as considered in e.g.~\cite{furon09b, laarhoven14ihmmsec}, but other decoders considered in~\cite{furon14} may be used as well.

After obtaining the ``evidence'' $\vc{x}_j, \vc{y}, \vc{p}$, the distributor wants to distinguish between whether user $j$ is guilty or not\footnote{Note that we only consider $\vc{x}_j$ to be part of the evidence, instead of $\mathcal{X}$. Using all of $\mathcal{X}$ for decoding would correspond to joint decoding; this is discussed later on.}. The Neyman-Pearson lemma tells us that the most powerful test to distinguish between $H_0$ and $H_1$ (minimizing one error probability, when the other is fixed) is to test whether the following likelihood ratio exceeds an appropriately chosen threshold $\eta'$. We write $f_{A}(a) = \pr(A = a)$ for random variables $A$.
\begin{align}
\Lambda(\vc{x}_j, \vc{y}, \vc{p}) = \frac{f_{\vc{X}_j, \vc{Y} | \vc{P}}(\vc{x}_j, \vc{y} | \vc{p}, H_1)}{f_{\vc{X}_j, \vc{Y} |\vc{P}}(\vc{x}_j, \vc{y} | \vc{p}, H_0)}.
\end{align}
Taking logarithms, and noting that different positions $i$ are i.i.d., testing whether a user's likelihood ratio exceeds $\eta'$ is equivalent to testing whether his score $S_j = \sum_i g(x_{j,i}, y_i. p_i)$ exceeds $\eta = \ln \eta'$ for $g$ defined as follows. Here we omit subscripts on $X$, $Y$ and $P$, as the random variables are i.i.d.\ for different $i,j$.
\begin{align}
g(x,y,p) = \ln\left(\frac{f_{X,Y|P}(x,y|p,H_1)}{f_{X,Y|P}(x,y|p,H_0)}\right). \label{eq:dec-simple-g}
\end{align}
Results of Abbe and Zheng~\cite{abbe12} have shown that in certain applications, a (generalized) linear decoder designed against the worst-case attack is asymptotically optimal. Since the worst-case attack for finite $c$ is somewhat hard to compute, but is known to be close to the interleaving attack~\cite{furon09c, huang09b} (and asymptotically equal to it~\cite{huang12}), an approximation of this optimal decoder may be obtained by assuming the colluders used the interleaving attack $\bm{\theta} = \thint$, defined by
\begin{align}
(\thint)_z = \frac{z}{c} \, . \qquad (0 \leq z \leq c)
\end{align}
In that case, working out the probabilities for fixed $c$ leads to the following score function $g$~\cite{laarhoven14ihmmsec}:
\begin{align}
g(x,y,p) = \begin{cases} 
\ln\left(1 + \frac{p}{c(1 - p)}\right) & \text{if } x = y = 0; \\
\ln\left(1 - \frac{1}{c}\right) & \text{if } x \neq y; \\
\ln\left(1 + \frac{1 - p}{c p}\right) & \text{if } x = y = 1. \end{cases} \label{eq:g}
\end{align}

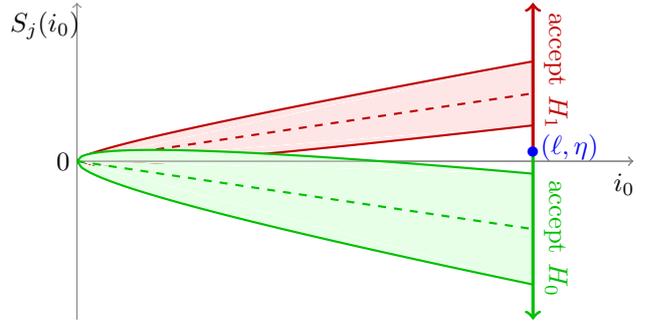
\begin{figure}[t]
\begin{tikzpicture}[scale = 0.6]
\draw [black!50!white, ->] (0, -3.5) -- (0, 3.5);
\draw [black!50!white, ->] (0, 0) -- (12.2, 0);
\node at (-0.7, 3) {$S_j(i_0)$};
\node at (12, -0.5) {$i_0$};
\node at (-0.3, 0) {$0$};
\shade [left color=red!10, right color=red!10, line width=0pt] 
	(0, 0) -- (10, 1.5) -- (10, 2.207) --
	plot [smooth, samples=100, domain=0:10] (\x, {0.15*\x + sqrt(0.05*\x)}) -- cycle;	
\shade [left color=red!10, right color=red!10, line width=0pt] 
	(0, 0) -- (10, 1.5) -- (10, 0.793) --
	plot [smooth, samples=100, domain=0:10] (\x, {0.15*\x - sqrt(0.05*\x)}) -- cycle;
\draw [red!75!black, samples=100, domain=0:10, thick] 
	plot(\x, {0.15*\x + sqrt(0.05*\x)});
\draw [red!75!black, samples=100, domain=0:10, thick, dashed] 
	plot(\x, {0.15*\x});
\draw [red!75!black, samples=100, domain=0:10, thick] 
	plot(\x, {0.15*\x - sqrt(0.05*\x)});
\shade [left color=green!10, right color=green!10] 
	(0, 0) -- (10, -1.5) -- (10, -0.275) --
	plot [smooth, samples=100, domain=0:10] (\x, {-0.15*\x + sqrt(0.15*\x)}) -- cycle;	
\shade [left color=green!10, right color=green!10] 
	(0, 0) -- (10, -1.5) -- (10, -2.725) --
	plot [smooth, samples=100, domain=0:10] (\x, {-0.15*\x - sqrt(0.15*\x)}) -- cycle;
\draw [black!50!white, ->] (0, 0) -- (12.2, 0);
\draw [green!75!black, samples=100, domain=0:10, thick] 
	plot(\x, {-0.15*\x + sqrt(0.15*\x)});
\draw [green!75!black, samples=100, domain=0:10, thick, dashed] 
	plot(\x, {-0.15*\x});
\draw [green!75!black, samples=100, domain=0:10, thick] 
	plot(\x, {-0.15*\x - sqrt(0.15*\x)});
\draw [very thick, red!75!black, ->] (10, 0.2) -- (10, 3.5);
\node [rotate=270] at (10.5, 2.0) {\color{red!75!black} accept $H_1$};
\draw [very thick, green!75!black, ->] (10, 0.2) -- (10, -3.5);
\node [rotate=270] at (10.5, -1.7) {\color{green!75!black} accept $H_0$};
\node at (10, 0.2) {\color{blue}\textbullet};
\node at (10.8, 0.3) {\color{blue} $(\ell, \eta)$};
\end{tikzpicture}
\caption{Tardos' scheme with log-likelihood decoding. The green and red marked areas (dashed lines) indicate the range (average) of innocent and guilty user scores respectively. Accepting $H_0$ or $H_1$ is based on whether $S_j(\ell) > \eta$, i.e., whether a user's score ends up above or below the blue point $(\ell, \eta)$. \label{fig:sketch1}}
\end{figure}

To sketch the situation of cumulative user scores and the accusation procedure, Figure~\ref{fig:sketch1} outlines the scores $S_j(i_0) = \sum_{i=1}^{i_0} S_{j,i}$ against $i_0$, for $i_0 = 0$ up to the final moment of decision $i_0 = \ell$. Assuming a colluder-symmetric collusion channel, scores of users $j \notin \mathcal{C}$ follow a certain random walk with a negative drift $\mu_0 < 0$ and a relatively large variance $\sigma_0^2$, while scores of guilty users $j \in \mathcal{C}$ follow a random walk with a positive drift $\mu_1 > 0$ and a smaller variance $\sigma_1^2$. 

\subsection{Sequential scheme}

The improvement described in~\cite{laarhoven13tit} for the adaptive setting does not change the code generation phase at all, so although it was coined the \textit{dynamic} Tardos scheme, it may more suitably be called the \textit{sequential} Tardos scheme. The modification compared to the non-adaptive scheme described above, to make better use of the sequential setting, is the following: instead of only cutting off users from the content at the very end, when their scores exceed $\eta$, we disconnect users as soon as their normalized scores exceed the normalized threshold $\eta$. This prevents the colluder from contributing to the remaining parts of the content, and allows the distributor to find the remaining colluders as well. Here by normalization we refer to translating the scores by $+\ell \mu_0$, so that innocent users are expected to have an average final score of $0$.

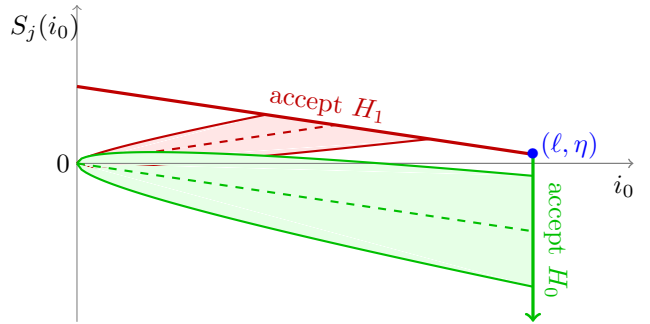
\begin{figure}[t]
\begin{tikzpicture}[scale = 0.6]
\draw [black!50!white, ->] (0, -3.5) -- (0, 3.5);
\draw [black!50!white, ->] (0, 0) -- (12.2, 0);
\node at (-0.7, 3.0) {$S_j(i_0)$};
\node at (12.0, -0.5) {$i_0$};
\node at (-0.3, 0) {$0$};
\shade [bottom color=red!10, top color=red!10, line width=0pt] 
	(0, 0) -- (5.667, 0.850) -- (4.149, 1.078) --
	plot [smooth, samples=100, domain=0:4.149] (\x, {0.15*\x + sqrt(0.05*\x)}) -- cycle;	
\shade [bottom color=red!10, top color=red!10, line width=0pt] 
	(0, 0) -- (5.667, 0.850) -- (7.740, 0.539) --
	plot [smooth, samples=100, domain=0:7.740] (\x, {0.15*\x - sqrt(0.05*\x)}) -- cycle;
\draw [red!75!black, samples=100, domain=0:4.149, thick] 
	plot(\x, {0.15*\x + sqrt(0.05*\x)});
\draw [red!75!black, samples=100, domain=0:5.667, thick, dashed] 
	plot(\x, {0.15*\x});
\draw [red!75!black, samples=100, domain=0:7.740, thick] 
	plot(\x, {0.15*\x - sqrt(0.05*\x)});
\shade [left color=green!10, right color=green!10] 
	(0, 0) -- (10, -1.5) -- (10, -0.275) --
	plot [smooth, samples=100, domain=0:10] (\x, {-0.15*\x + sqrt(0.15*\x)}) -- cycle;	
\shade [left color=green!10, right color=green!10] 
	(0, 0) -- (10, -1.5) -- (10, -2.725) --
	plot [smooth, samples=100, domain=0:10] (\x, {-0.15*\x - sqrt(0.15*\x)}) -- cycle;
\draw [black!50!white, ->] (0, 0) -- (12.2, 0);
\draw [green!75!black, samples=100, domain=0:10, thick] 
	plot(\x, {-0.15*\x + sqrt(0.15*\x)});
\draw [green!75!black, samples=100, domain=0:10, thick, dashed] 
	plot(\x, {-0.15*\x});
\draw [green!75!black, samples=100, domain=0:10, thick] 
	plot(\x, {-0.15*\x - sqrt(0.15*\x)});
\draw [very thick, red!75!black] (10, 0.2) -- (0, 1.7);
\node [rotate=-8.5] at (5.5, 1.3) {\color{red!75!black} accept $H_1$};
\draw [very thick, green!75!black, ->] (10, 0.2) -- (10, -3.5);
\node [rotate=270] at (10.5, -1.7) {\color{green!75!black} accept $H_0$};
\node at (10, 0.2) {\color{blue}\textbullet};
\node at (10.8, 0.4) {\color{blue} $(\ell, \eta)$};
\end{tikzpicture}
\caption{Laarhoven et al.'s sequential Tardos scheme with log-likelihood scores. The red decreasing line, which runs parallel to the dashed green line, shows when users are accused and disconnected. \label{fig:sketch2}}
\end{figure}

To illustrate the effect of this change to the scheme, Figure~\ref{fig:sketch2} sketches the cumulative user scores in the sequential setting \textit{without normalization}, and the new accusation criterion. Without normalization, the scores follow the same general path as in Figure~\ref{fig:sketch1}, and the red accusation threshold becomes a decreasing line, rather than a horizontal line as in~\cite{laarhoven13tit, laarhoven13wifs}. As discussed in~\cite{laarhoven13tit}, with this modification one can provably find all colluders rather than only one with a similar provable code length as in the non-adaptive setting. The central result of~\cite{laarhoven13tit} can be stated as follows.

\begin{theorem} \cite{laarhoven13tit} \label{thm:adaptive1}
Suppose $\ell$ and $\eta$ are chosen in the non-adaptive Tardos scheme to guarantee that
\begin{itemize}
  \item[(i)] with prob.\ at least $1 - \eps_1$ no innocent users are accused;
  \item[(ii)] with prob.\ at least $1 - \eps_2$ at least one colluder is accused.
\end{itemize}
Then, using almost the same scheme parameters as before\footnote{This disregards a small technical detail regarding the overshoot over the boundary $\eta$; see the discussion of $Z$ and $\tilde{Z}$ in \cite[Section III.C]{laarhoven13tit}. To be sure that the scheme still works we can disregard scores right after a user is removed from the system~\cite[Section II]{laarhoven13wifs} with a negligible increase in $\ell$. We omit details here, and only present the simplified result.}, with this sequential construction we can guarantee that
\begin{itemize}
  \item[(i)] with pr.\ at least $1 - 2 \eps_1$ no innocent users are accused;
  \item[(ii)] with prob.\ at least $1 - 2 \eps_2$ all colluders are accused.
\end{itemize}
\end{theorem}

In practice, this means that to turn a non-adaptive scheme into a sequential scheme that provably finds all colluders, we just have to replace $\eps_1$ and $\eps_2$ by $\frac{1}{2} \eps_1$ and $\frac{1}{2} \eps_2$ in the formulas for $\ell$ and $\eta$ of the non-adaptive setting. Since $\ell$ only depends logarithmically on $\eps_1$ and $\eps_2$, for large $n$ and $c$ the resulting increase in the code length is negligible.

\subsection{Sequential variants}

While the above sequential scheme deals well with the setting where $c$ is known and users can be accused after every position $i$, the paper~\cite{laarhoven13tit} also discussed slight variations of this setting, which may well appear in practice. In particular, the two problems of not being able to cut off users after every segment $i$, and not knowing $c$, were addressed in~\cite[Sections IV and V]{laarhoven13tit}.

\subsubsection{Weakly sequential decoding}

To make tracing harder, pirates may delay the pirate output, so that a user whose score exceeds $\eta$ at time $i_0$ can only be disconnected at time, say, $i_0 + B$. As we are now quite certain that he is guilty, and since he contributed to segments $i_0 + 1, \dots, i_0 + B$, we could consider these segments \textit{tainted} and disregard them completely for tracing the remaining colluders. This solution was proposed in~\cite[Section IV.A]{laarhoven13tit} and it was shown to lead to a moderate increase in the code length of $(c - 1)B$. A different analysis in~\cite[Section IV.B]{laarhoven13tit} showed that one can also perform a new study of the possible overshoot over the boundary $\eta$, due to the increase $B$, leading to a higher increase in the code length. Therefore the solution from~\cite[Sect.~IV.A]{laarhoven13tit} should be preferred.

\subsubsection{Universal sequential decoding}

As for the setting where $c$ is unknown and only a crude estimate $c_0$ is known (or no bound is known at all), \cite[Section V]{laarhoven13tit} proposed a method where each user is assigned several scores $S_j^{(1)}, \dots, S_j^{(c_0)}$ based on how large the coalition is estimated to be, and disconnecting a user as soon as one of his scores crosses one of the corresponding boundaries $\eta^{(1)}, \dots, \eta^{(c_0)}$. It was noted in~\cite{laarhoven13tit} that the scores are very similar and the boundaries seem to correspond to a continuous function $\eta(i_0) \propto \sqrt{i_0}$. One of the open problems posed in~\cite[Section VII.B]{laarhoven13tit} was therefore whether schemes with single scores and curved boundaries are provably secure.

\subsubsection{Joint decoding}

Finally, another topic often considered in the fingerprinting literature is joint decoding~\cite{amiri09, berchtold12, charpentier09, huang12, meerwald11, meerwald12, moulin08, oosterwijk14, skoric14}: using the entire code $\mathcal{X}$, rather than only the user's code word $\vc{x}_j$, to decide whether user $j$ should be accused. Assigning scores to tuples of users was considered before in e.g.\ \cite{oosterwijk14}, but no explicit decision criterion with provable results was provided, and it was left as an open problem.

\section{Wald's scheme}
\label{sec:wald}

\subsection{Sequential scheme}

To understand the motivation behind the sequential Tardos scheme, and to see how the design can possibly be improved, we now turn our attention to what has long been known in statistics literature to be a solution for hypothesis testing in sequential settings: Wald's sequential probability ratio test (SPRT). This scheme originated in the 1940s~\cite{wald45, wald47}, and countless follow-up works have appeared since, which have been summarized in many books on this topic~\cite{bartroff13, chernoff72, govindarajulu04, jennison00, mukhopadhyay09, siegmund85, wald47, wetherill86}. 

Let us recall the formulation of the fingerprinting problem in terms of hypothesis testing, where we want to distinguish between the following two hypotheses:
\begin{itemize}
	\item $H_0$: user $j$ is innocent ($j \notin \mathcal{C}$).
	\item $H_1$: user $j$ is guilty ($j \in \mathcal{C}$).
\end{itemize}
Now, to decide between these two hypotheses in sequential settings, Wald proposed the following procedure. Let $\eta_1$ and $\eta_0$ be two constants, with $\eta_1 > 0 > \eta_0$, and again let us use the optimal log-likelihood score function $g$ from~\eqref{eq:g}. Now we decide in favor of $H_1$ as soon as a user's cumulative score exceeds $\eta_1$, and we decide to accept $H_0$ as soon as the user's score drops below $\eta_0$. As long as a user score stays in the interval $[\eta_0, \eta_1]$, we continue testing. This accusation procedure is sketched in Figure~\ref{fig:sketch3}. 

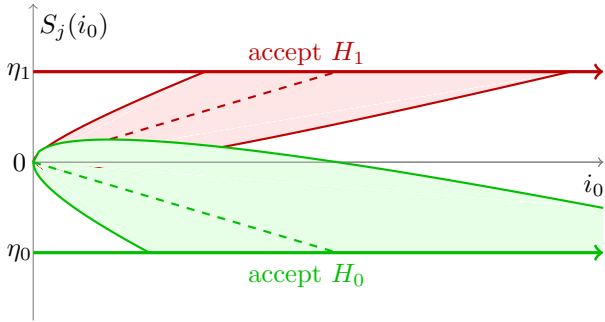
\begin{figure}[t]
\begin{tikzpicture}[scale = 0.6]
\draw [black!50!white, ->] (0, -3.5) -- (0, 3.5);
\draw [black!50!white, ->] (0, 0) -- (12.5, 0);
\node at (0.9, 3.0) {$S_j(i_0)$};
\node at (12.3, -0.5) {$i_0$};
\node at (-0.3, 0) {$0$};
\node at (-0.3, 2) {$\eta_1$};
\node at (-0.3, -2) {$\eta_0$};
\shade [left color=red!10, right color=red!10, line width=0pt] 
	(0, 0) -- (6.667, 2) -- (3.772, 2) --
	plot [smooth, samples=100, domain=0:3.772] (\x, {0.3*\x + sqrt(0.2*\x)}) -- cycle;	
\shade [left color=red!10, right color=red!10, line width=0pt] 
	(0, 0) -- (6.667, 2) -- (11.784, 2) --
	plot [smooth, samples=100, domain=0:11.784] (\x, {0.3*\x - sqrt(0.2*\x)}) -- cycle;
\draw [red!75!black, samples=100, domain=0:3.772, thick] 
	plot(\x, {0.3*\x + sqrt(0.2*\x)});
\draw [red!75!black, samples=100, domain=0:6.667, thick, dashed] 
	plot(\x, {0.3*\x});
\draw [red!75!black, samples=100, domain=0:11.784, thick] 
	plot(\x, {0.3*\x - sqrt(0.2*\x)});
\shade [left color=green!10, right color=green!10] 
	(0, 0) -- (6.667, -2) -- (12.5, -2) -- (12.5, -1.015) --
	plot [smooth, samples=100, domain=0:12.5] (\x, {-0.3*\x + sqrt(0.6*\x)}) -- cycle;
\shade [left color=green!10, right color=green!10] 
	(0, 0) -- (6.667, -2) -- (2.546, -2) --
	plot [smooth, samples=100, domain=0:2.546] (\x, {-0.3*\x - sqrt(0.6*\x)}) -- cycle;	
\draw [black!50!white, ->] (0, 0) -- (12.5, 0);
\draw [green!75!black, samples=100, domain=0:12.50, thick] 
	plot(\x, {-0.3*\x + sqrt(0.6*\x)});
\draw [green!75!black, samples=100, domain=0:6.667, thick, dashed] 
	plot(\x, {-0.3*\x});
\draw [green!75!black, samples=100, domain=0:2.546, thick] 
	plot(\x, {-0.3*\x - sqrt(0.6*\x)});
\draw [very thick, red!75!black, ->] (0, 2.0) -- (12.5, 2.0);
\node at (6, 2.4) {\color{red!75!black} accept $H_1$};
\draw [very thick, green!75!black, ->] (0, -2.0) -- (12.5, -2.0);
\node at (6, -2.5) {\color{green!75!black} accept $H_0$};
\end{tikzpicture}
\caption{Wald's SPRT construction. As soon as a user's score leaves the interval $[\eta_0, \eta_1]$, he is marked innocent (below $\eta_0$) or guilty (above $\eta_1$). \label{fig:sketch3}}
\end{figure}

\textbf{Choosing the thresholds.} To understand how the parameters $\eta_0$ and $\eta_1$ should be chosen, a connection is often made with the continuous-time analog of random walks, Brownian motions. Assuming that user scores are continuous, so that when a score crosses one of the boundaries it really \textit{hits} the boundary (rather than jumping over it, in the discrete model), then to guarantee that an innocent user is acquitted with probability at least $1 - \eps_1'$ and a guilty user is accused with probability at least $1 - \eps_2'$, the following choice is optimal:
\begin{align}
\eta_0 = \ln\left(\frac{\eps_2'}{1 - \eps_1'}\right), \qquad \eta_1 = \ln\left(\frac{1 - \eps_2'}{\eps_1'}\right). \label{eq:eta-agg}
\end{align}
To guarantee that all innocent users are acquitted and all guilty users are found, we need to let $\eps_1' = O(\tfrac{1}{n})$ and $\eps_2' = O(\tfrac{1}{c})$, which for large $c,n$ means $\eta_0 \sim -\ln c$ and $\eta_1 \sim \ln n$. For instance, writing $\eps_1' = \eps_1 / n$ and $\eps_2' = \eps_2/c$, so that the probability of not accusing innocents (accusing all guilties) is at least $1 - \eps_1$ ($1 - \eps_2$), this corresponds to taking
\begin{align}
\eta_0 = \ln\left(\frac{\eps_2/c}{1 - \eps_1/n}\right), \qquad \eta_1 = \ln\left(\frac{1 - \eps_2/c}{\eps_1/n}\right).
\end{align}

There are two important issues that we need to address, the first of which is that we are not dealing with continuous user scores but discrete scores. One of the effects of having discrete jumps in the scores is that there may be a slight \textit{overshoot} over one of the boundaries when a user is accused or acquitted; a score may cross one of the lines at a non-integral point so to say, and at the next measurement the score may significantly exceed $\eta_1$ or drop below $\eta_0$. As a result the error probabilities for the above choice of thresholds are not exact. A useful property of the above choice of parameters is that if by $\tilde{\eps}_1'$ and $\tilde{\eps}_2'$ we denote the \textit{real} probabilities of accusing innocent and guilty users, when using these thresholds $\eta_0$ and $\eta_1$, we have~\cite[Equation (3.30)]{wald45}
\begin{align}
\tilde{\eps}_1' + \tilde{\eps}_2' \leq \eps_1' + \eps_2'.
\end{align}
In other words, the total error probability does not increase, and at most one of $\eps_1'$ and $\eps_2'$ might increase. Alternatively, exact bounds on the error probabilities can be obtained, showing that the following slightly conservative choice of parameters guarantees that the error bounds are satisfied:
\begin{align}
\eta_0 = -\ln\left(1/\eps_2'\right), \qquad \eta_1 = \ln\left(1/\eps_1'\right). \label{eq:eta-cons}
\end{align}

The second issue that we should address is that having a threshold $\eta_0$ only makes sense if all colluders have an increasing score. If the colluders know about the tracing algorithm, and use an asymmetric pirate strategy, e.g.\ by letting one colluder be inactive at the start and letting him join in later, this colluder will incorrectly be acquitted early on. In this setting one could say that innocence is virtually impossibly to prove, while it is possible to prove that someone is guilty. To deal with this problem, a simple solution is not to use a lower threshold $\eta_0$ at all. This is equivalent to setting $\eps_2' = 0$, as that way we will never acquit a colluder. In that case, the conservative choice of thresholds from~\eqref{eq:eta-agg} can be stated as
\begin{align}
\eta_0 = -\infty, \qquad \eta_1 = \ln\left(1/\eps_1'\right). \label{eq:eta-no}
\end{align}
Note that in this case, the aggressive and conservative expressions from~\eqref{eq:eta-agg} and \eqref{eq:eta-cons} match, i.e., $\eps_1'$ is a tight bound on the probability of incorrectly accusing a single innocent user. This more realistic implementation of the sequential probability ratio test in the uninformed fingerprinting game is again sketched in Figure~\ref{fig:sketch4}.

\begin{figure}[t]
\begin{tikzpicture}[scale = 0.6]
\draw [black!50!white, ->] (0, -3.5) -- (0, 3.5);
\draw [black!50!white, ->] (0, 0) -- (12.5, 0);
\node at (0.9, 3) {$S_j(i_0)$};
\node at (12.3, -0.5) {$i_0$};
\node at (-0.3, 0) {$0$};
\node at (-0.3, 2) {$\eta_1$};
\shade [left color=red!10, right color=red!10, line width=0pt] 
	(0, 0) -- (6.667, 2) -- (3.772, 2) --
	plot [smooth, samples=100, domain=0:3.772] (\x, {0.3*\x + sqrt(0.2*\x)}) -- cycle;	
\shade [left color=red!10, right color=red!10, line width=0pt] 
	(0, 0) -- (6.667, 2) -- (11.784, 2) --
	plot [smooth, samples=100, domain=0:11.784] (\x, {0.3*\x - sqrt(0.2*\x)}) -- cycle;

\draw [red!75!black, samples=100, domain=0:3.772, thick] 
	plot(\x, {0.3*\x + sqrt(0.2*\x)});
\draw [red!75!black, samples=100, domain=0:6.667, thick, dashed] 
	plot(\x, {0.3*\x});
\draw [red!75!black, samples=100, domain=0:11.784, thick] 
	plot(\x, {0.3*\x - sqrt(0.2*\x)});
\shade [left color=green!10, right color=green!10] 
	(0, 0) -- (11.667, -3.5) -- (12.5, -3.5) -- (12.5, -1.011) --
	plot [smooth, samples=100, domain=0:12.5] (\x, {-0.3*\x + sqrt(0.6*\x)}) -- cycle;
\shade [left color=green!10, right color=green!10] 
	(0, 0) -- (11.667, -3.5) -- (5.572, -3.5) --
	plot [smooth, samples=100, domain=0:5.572] (\x, {-0.3*\x - sqrt(0.6*\x)}) -- cycle;	
\draw [black!50!white, ->] (0, 0) -- (12.5, 0);
\draw [green!75!black, samples=100, domain=0:12.50, thick] 
	plot(\x, {-0.3*\x + sqrt(0.6*\x)});
\draw [green!75!black, samples=100, domain=0:11.667, thick, dashed] 
	plot(\x, {-0.3*\x});
\draw [green!75!black, samples=100, domain=0:5.572, thick] 
	plot(\x, {-0.3*\x - sqrt(0.6*\x)});
\draw [very thick, red!75!black, ->] (0, 2.0) -- (12.5, 2.0);
\node at (6, 2.4) {\color{red!75!black} accept $H_1$};
\end{tikzpicture}
\caption{Wald's SPRT with no early innocent decisions. Setting $\eta_1 = \ln(1/\eps_1')$ guarantees that innocent users are never accused with probability at least $1 - \eps_1'$, while this design guarantees that $\eps_2 = 0$. \label{fig:sketch4}}
\end{figure}
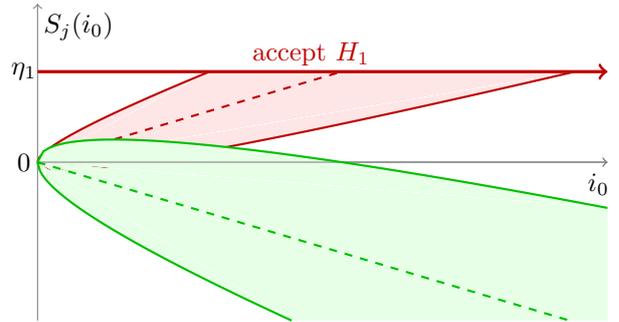

\textbf{Optimality of the SPRT.} Although reaching a decision with this procedure may theoretically take a very long time, Wald proved that his test procedure always terminates~\cite[Appendix~A]{wald47}, regardless of $\eps_1$ and $\eps_2$. Furthermore, if by $\mu_0$ ($\mu_1$) and $\sigma_0^2$ ($\sigma_1^2$) we denote the expected score in one segment for innocent (guilty) users, then we know that with high probability, the procedure will terminate not long after $i_0 \cdot \mu_0 + O(\sigma_0)$ ($i_0 \cdot \mu_1 + O(\sigma_1)$) crosses the boundary $\eta_0$ ($\eta_1$).

More formally, Wald analyzed the expected time by which his procedure terminates, under either $H_0$ or $H_1$, and together with Wolfowitz he proved~\cite{wald48} that his SPRT is optimal in that it minimizes the expected time before a decision is reached, both under $H_0$ and under $H_1$. Ignoring overshoots over the boundary (i.e., assuming we are dealing with continuous random walks), he further derived explicit expressions for both these expected termination times, which are stated below. In the following theorem, we write $\kl(a \| b) = a \ln(\frac{a}{b}) + (1 - a) \ln(\frac{1-a}{1-b})$ for the Kullback-Leibler divergence or relative entropy (in nats) between $a$ and $b$.

\begin{theorem} \cite{wald47, wald48} \label{thm:adaptive2}
Suppose we have a sequential test procedure, for which
\begin{itemize}
	\item an innocent user is accused w.p. at most $\eps_1'$;
	\item a guilty user is acquitted w.p. at most $\eps_2'$;
	\item the probability of termination is $1$.
\end{itemize}
Let $T$ denote the time at which a decision is reached. Then:
\begin{align}
\expn(T | H_0) &\geq \frac{1}{-\mu_0} \kl(\eps_1' \| 1 - \eps_2') \approx \frac{\ln(1/\eps_2')}{-\mu_0} \, , \\
\expn(T | H_1) &\geq \frac{1}{\mu_1} \kl(\eps_2' \| 1 - \eps_1') \approx \frac{\ln(1/\eps_1')}{\mu_1} \, .
\end{align}
Furthermore, the sequential probability ratio test is a sequential test simultaneously minimizing both $\expn(T_0)$ and $\expn(T_1)$, and assuming that there is no overshoot over the boundaries, both inequalities above are equalities for the SPRT. 
\end{theorem}

For large $n$, the per-user false positive error probability scales as $\eps_1' = \Theta(1/n)$ while $\eps_2' = \Theta(1)$ based on the argument that if the average pirate score exceeds $\eta_1$, all pirate scores exceed $\eta_1$~\cite{laarhoven13tit}. We further have that
\begin{align}
\mu_1 &= \expn(S_{j,i} | H_1) \\
&= \expn_P \sum_{x,y} f_{X,Y|P}(x,y|p,H_1) \ln \frac{f_{X,Y|P}(x,y|p,H_1)}{f_{X,Y|P}(x,y|p,H_0)} \\
&= (\ln 2) \expn_P I(X_1; Y|P),
\end{align}
where $I(X_1, Y|P = p)$ is the mutual information between a pirate symbol and the pirate output, as described in e.g.\ \cite{huang12, laarhoven14ihmmsec}. This leads to the following corollary.

\begin{theorem} \label{thm:adaptive3}
For sequential tests satisfying the conditions stated in Theorem~\ref{thm:adaptive2}, we have:
\begin{align}
\expn(T | H_1) &\gtrsim \frac{\log_2 n}{\expn_P I(X_1; Y|P)} \, .
\end{align}
\end{theorem} 

This result implies that in general, sequentiality does not lead to a decrease in the asymptotic code length; with non-adaptive schemes it is also possible to achieve this asymptotic code length~\cite{laarhoven14ihmmsec}. The two gains of sequential testing are that (i) in fact \textit{all} colluders, rather than at least one of them, can provably be caught with this asymptotic code length; and (ii) in practice, for finite $c$ and $n$, the time needed to find and trace all colluders will generally be shorter than in the non-adaptive setting. Although the asymptotic code length are the same, the convergence to this limit is significantly faster for sequential schemes than for non-adaptive schemes.
 
While most of the analyses and results above are based on running this scheme with parallel infinite boundaries, it is not impossible to force an early decision. As already described by Wald~\cite[Section 3.8]{wald47}, one might ultimately prefer to \textit{truncate} the test procedure at some fixed time $\ell$, at which we make a decision similar to the sequential Tardos scheme, and similar to the non-adaptive setting. This may be done with and without a lower boundary; a sketch for the case with a lower boundary is given in Figure~\ref{fig:sketch5}. Analyzing these variants rigorously seems difficult, even with Brownian approximations, but an interested reader may refer to e.g.\ one of the books on sequential testing listed at the beginning of this section. With truncation, one should ask the question whether forcing a decision by some fixed time $\ell$ is really important. After all, if the main goal is to minimize the \textit{worst-case} code length $\ell$ needed to make a decision, then it is commonly best to wait until the very end and to take all evidence into account before making any decisions at all; which exactly corresponds to the non-adaptive setting. 

\begin{figure}[t]
\begin{tikzpicture}[scale = 0.6]
\draw [black!50!white, ->] (0, -3.5) -- (0, 3.5);
\draw [black!50!white, ->] (0, 0) -- (12.2, 0);
\node at (0.9, 3) {$S_j(i_0)$};
\node at (12.0, -0.5) {$i_0$};
\node at (-0.3, 0) {$0$};
\node at (-0.3, 2) {$\eta_1$};
\node at (-0.3, -2) {$\eta_0$};
\shade [left color=red!10, right color=red!10, line width=0pt] 
	(0, 0) -- (6.667, 2) -- (3.772, 2) --
	plot [smooth, samples=100, domain=0:3.772] (\x, {0.3*\x + sqrt(0.2*\x)}) -- cycle;	
\shade [left color=red!10, right color=red!10, line width=0pt] 
	(0, 0) -- (6.667, 2) -- (10, 2) -- (10, 1.586) -- 
	plot [smooth, samples=100, domain=0:10] (\x, {0.3*\x - sqrt(0.2*\x)}) -- cycle;
\draw [red!75!black, samples=100, domain=0:3.772, thick] 
	plot(\x, {0.3*\x + sqrt(0.2*\x)});
\draw [red!75!black, samples=100, domain=0:6.667, thick, dashed] 
	plot(\x, {0.3*\x});
\draw [red!75!black, samples=100, domain=0:10.00, thick] 
	plot(\x, {0.3*\x - sqrt(0.2*\x)});
\shade [left color=green!10, right color=green!10] 
	(0, 0) -- (6.667, -2) -- (10, -2) -- (10, -0.551) --
	plot [smooth, samples=100, domain=0:10] (\x, {-0.3*\x + sqrt(0.6*\x)}) -- cycle;
\shade [left color=green!10, right color=green!10] 
	(0, 0) -- (6.667, -2) -- (2.546, -2) --
	plot [smooth, samples=100, domain=0:2.546] (\x, {-0.3*\x - sqrt(0.6*\x)}) -- cycle;	
\draw [black!50!white, ->] (0, 0) -- (12.2, 0);
\draw [green!75!black, samples=100, domain=0:10.00, thick] 
	plot(\x, {-0.3*\x + sqrt(0.6*\x)});
\draw [green!75!black, samples=100, domain=0:6.667, thick, dashed] 
	plot(\x, {-0.3*\x});
\draw [green!75!black, samples=100, domain=0:2.546, thick] 
	plot(\x, {-0.3*\x - sqrt(0.6*\x)});
\draw [very thick, red!75!black] (0, 2.0) -- (10, 2.0);
\draw [very thick, red!75!black] (10, 2.0) -- (10, 0.5);
\node at (6, 2.4) {\color{red!75!black} accept $H_1$};
\draw [very thick, green!75!black] (10, 0.5) -- (10, -2.0);
\draw [very thick, green!75!black] (10, -2.0) -- (0, -2.0);
\node at (6, -2.5) {\color{green!75!black} accept $H_0$};
\node at (10, 0.5) {\color{blue}\textbullet};
\node at (10.8, 0.5) {\color{blue} $(\ell, \eta)$};
\end{tikzpicture}
\caption{Wald's SPRT with truncated thresholds, guaranteeing a decision after at most $\ell$ segments. For small $\ell$, both $\eps_1$ and $\eps_2$ will increase. \label{fig:sketch5}}
\end{figure}
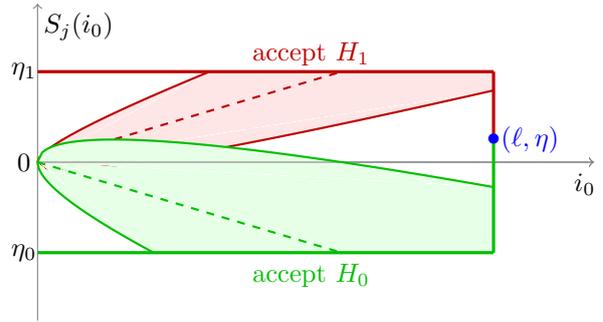

\subsection{Sequential variants}

The SPRT has received extensive attention in the literature, with thorough analyses of the effects of the overshoot over the boundaries, slight modifications of the scheme (such as the truncated SPRT mentioned above), and the effects of using different boundaries than the horizontal lines in the figures above. We highlight two variants which we also considered for the sequential Tardos scheme, and we consider how joint decoding may be done with the SPRT. For further details we refer the interested reader to e.g.~\cite{bartroff13, chernoff72, govindarajulu04, jennison00, mukhopadhyay09, siegmund85, wald47, wetherill86}.

\subsubsection{Weakly sequential decoding}

In the setting of weakly adaptive decoding, where pirates delay their rebroadcast of the content (or where content is sent out in blocks of size $B$), the results based on continuous approximations using Brownian motions become less and less accurate. For higher values of $B$, the overshoot over the boundary becomes more and more significant, which was also discussed in~\cite[Section IV.B]{laarhoven13tit}; there the possible overshoot was parametrized by $\tilde{Z}_B - Z$, and it was noted that exactly this overshoot causes problems.

To deal with this problem effectively, we can again use the method described in~\cite[Section IV.A]{laarhoven13tit}: ignore the tainted segments $i$ to which a user who is now deemed guilty may have contributed. Then the increase in the code length may again only be $(c-1)B$, which in the uninformed fingerprinting game is negligible with respect to $\ell \propto c^2 \log n$.

\subsubsection{Universal decoding}

Recall that in the universal decoding setting, we assume that $c$ is unknown, and only a crude bound $c_0 \gg c$ may be known. To deal with this, Laarhoven et al.~\cite{laarhoven13tit} proposed to keep multiple scores per user, and multiple accusation boundaries. It was conjectured that using a single score for each user, and using a curved boundary function of the form $\eta_1(i_0) \propto \sqrt{i_0}$ may be possible.

In terms of hypothesis testing, testing whether $j \in \mathcal{C}$ or $j \notin \mathcal{C}$ for unknown coalition sizes $c$ could be considered a test of a simple null hypothesis $H_0: \mu = \mu_0$ against a one-sided alternative $H_1: \mu > \mu_0$. In the informed setting, where the collusion channel is known, we might know exactly what $\mu_0$ is, and so such a one-sided test may form a solution. In that case, curved stopping boundaries (in particular, having a boundary of the shape $\eta_1(i_0) \propto \sqrt{i_0}$) has been suggested before; see e.g.~\cite[Chapter IV]{siegmund85}. When using the symmetric Tardos score function rather than the optimized log-likelihood ratios or MAP decoders, this approach may work well, although the issue remains that it seems that no single encoder and decoder can be used for arbitrary $c$ and $\boldsymbol{\theta}$: in all known cases, either the decoder depends on $c$ or $c_0$, or the encoder uses a cutoff which depends on $c$. 

To work with different score functions than Tardos' score function~\cite{tardos08} and \v{S}kori\'{c} et al.'s symmetric score function~\cite{skoric08}, where $\mu_0$ may be considered fixed, we need to circumvent the issue that $\mu_0$ may depend on $c$ and the pirate strategy $\boldsymbol{\theta}$ as well. In the universal uninformed decoding setting we therefore do not even know what $\mu_0$ is. Two hypotheses that may be more realistic to consider are $H_0: \mu \leq 0$ against $H_1: \mu > 0$: an innocent user will have a negative average score, while a guilty user will have a positive average score. However, depending on the collusion strategy, the values of $\mu_0$ and $\mu_1$ may both be small or large. This does not really help the colluders, as decreasing $|\mu_0|$ and $|\mu_1|$ also leads to a decrease in the variance of the scores, but it makes using a single linear decoder even more problematic.

To deal with these problems, the best solution for the universal setting may be to use a generalized linear decoder~\cite{abbe12, desoubeaux13, meerwald12}, and to normalize the scores during the decoding phase, as described in~\cite{laarhoven13wifs}. A generalized linear decoder is better suited for the setting of unknown $c$, and by normalizing user scores (which can be done based on $\mathcal{X}, \vc{p}, \vc{y}$) we know what $\mu_0$ is. Then we can again use a hypothesis test of the form $H_0: \mu = \mu_0$ against $H_1: \mu > \mu_0$ as described above, where a curved boundary may be optimal. 

\subsubsection{Joint decoding}

Recall that in joint decoding the entire code $\mathcal{X}$ is taken into account to decide whether users should be accused. In~\cite{oosterwijk14} it was considered to assign scores to tuples $\mathcal{T}$ of $t \leq c$ users, after which one would like to distinguish between the following $t + 1$ hypotheses:
\begin{itemize}
	\item $H_0$: tuple $\mathcal{T}$ contains no guilty users: $|\mathcal{T} \cap \mathcal{C}| = 0$;
	\item $H_1$: tuple $\mathcal{T}$ contains one guilty user: $|\mathcal{T} \cap \mathcal{C}| = 1$;
	\item $\dots$
	\item $H_t$: tuple $\mathcal{T}$ contains $t$ guilty users: $|\mathcal{T} \cap \mathcal{C}| = t$.
\end{itemize}
Although not quite as well studied as the case of two hypotheses, this topic has also received attention in SPRT literature, with the earliest work dating back to Sobel and Wald from 1949~\cite{sobel49}. They considered the problem of deciding between three simple hypotheses ($t = 2$), and provided a solution as sketched in Figure~\ref{fig:sketch6}. Using joint Neyman-Pearson decoder to assign scores to pairs of users, they considered the use of several stopping boundaries, each corresponding to a decision of accepting one of the hypotheses. The distribution of scores then depends on whether the tuple contains $0, 1$ or $2$ colluders, as illustrated by the green, yellow, and red highlighted curves in Figure~\ref{fig:sketch6}. This procedure can be generalized to multiple hypotheses as well, to deal with joint decoding with $c > 2$. For details on how to choose these stopping boundaries, see e.g.\ \cite{sobel49}.

\begin{figure}[t]
\begin{tikzpicture}[scale = 0.6]
\draw [black!50!white, ->] (0, -4) -- (0, 4);
\draw [black!50!white, ->] (0, 0) -- (12.2, 0);
\node at (0.9, 3.5) {$S_j(i_0)$};
\node at (12, -0.5) {$i_0$};
\node at (-0.3, 0) {$0$};
\node at (-0.3, 1) {$\eta_2$};
\node at (-0.3, -1) {$\eta_0$};
\shade [left color=red!10, right color=red!10, line width=0pt] 
	(0, 0) -- (3.333, 2) -- (1.111, 1.333) --
	plot [smooth, samples=100, domain=0:1.111] (\x, {0.6*\x + sqrt(0.4*\x)}) -- cycle;	
\shade [left color=red!10, right color=red!10, line width=0pt] 
	(0, 0) -- (3.333, 2) -- (10, 4) --
	plot [smooth, samples=100, domain=0:10] (\x, {0.6*\x - sqrt(0.4*\x)}) -- cycle;
\draw [red!75!black, samples=100, domain=0:1.111, thick] 
	plot(\x, {0.6*\x + sqrt(0.4*\x)});
\draw [red!75!black, samples=100, domain=0:3.333, thick, dashed] 
	plot(\x, {0.6*\x});
\draw [red!75!black, samples=100, domain=0:10.00, thick] 
	plot(\x, {0.6*\x - sqrt(0.4*\x)});
\shade [left color=green!10, right color=green!10] 
	(0, 0) -- (3.333, -2) -- (10, -4) --
	plot [smooth, samples=100, domain=0:10] (\x, {-0.6*\x + sqrt(0.4*\x)}) -- cycle;
\shade [left color=green!10, right color=green!10] 
	(0, 0) -- (3.333, -2) -- (1.111, -1.333) --
	plot [smooth, samples=100, domain=0:1.111] (\x, {-0.6*\x - sqrt(0.4*\x)}) -- cycle;	
\draw [green!75!black, samples=100, domain=0:10.00, thick] 
	plot(\x, {-0.6*\x + sqrt(0.4*\x)});
\draw [green!75!black, samples=100, domain=0:3.333, thick, dashed] 
	plot(\x, {-0.6*\x});
\draw [green!75!black, samples=100, domain=0:1.111, thick] 
	plot(\x, {-0.6*\x - sqrt(0.4*\x)});
\shade [left color=yellow!20, right color=yellow!20] 
	(0, 0) -- (3.333, 0) -- (10, 2) -- 
	plot [smooth, samples=100, domain=0:10] (\x, {sqrt(0.4*\x)}) -- cycle;
\shade [left color=yellow!20, right color=yellow!20] 
	(0, 0) -- (3.333, 0) -- (10, -2) --
	plot [smooth, samples=100, domain=0:10] (\x, {-sqrt(0.4*\x)}) -- cycle;	
\draw [black!50!white, ->] (0, 0) -- (12.2, 0);
\draw [yellow!50!black, samples=100, domain=0:9.7, thick] 
	plot(\x, {sqrt(0.4*\x)});
\draw [yellow!50!black, samples=100, domain=0:9.7, thick] 
	plot(\x, {-sqrt(0.4*\x)});
\draw [yellow!50!black, samples=100, domain=0:3.333, thick, dashed] 
	plot(\x, {0});
\draw [very thick, red!75!black, ->] (0, 1) -- (10, 4);
\node [rotate=17] at (6, 3.2) {\color{red!75!black} accept $H_2$};
\draw [very thick, yellow!50!black, ->] (3.333, 0) -- (12.2, 2.75);
\draw [very thick, yellow!50!black, ->] (3.333, 0) -- (12.2, -2.75);
\node [rotate=17] at (8, 1) {\color{yellow!50!black} accept $H_1$};
\node [rotate=-17] at (8, -1) {\color{yellow!50!black} accept $H_1$};
\draw [very thick, green!75!black, ->] (0, -1) -- (10, -4);
\node [rotate=-17] at (6, -3.2) {\color{green!75!black} accept $H_0$};
\end{tikzpicture}
\caption{Sobel and Wald's decision procedure to decide between three hypotheses $H_0, H_1, H_2$. Letting $H_i$ denote the event that a pair of users contains $i$ colluders, this leads to a joint decoding method. \label{fig:sketch6}}
\end{figure}
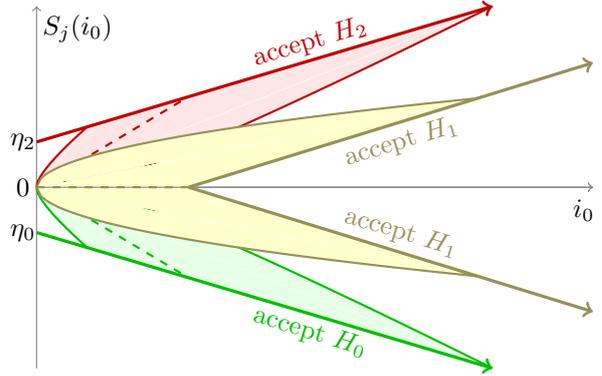

\section{Wald vs. Tardos: Applications}
\label{sec:comp}

In the previous two sections we saw how to construct sequential schemes based on the (dynamic) Tardos scheme, and based on Wald's SPRT. Here we briefly consider possible applications of both schemes, and how the two schemes compare. We consider three scenarios as follows:
\begin{enumerate} \vspace{-0.1cm}
	\item Defending against the interleaving attack. \vspace{-0.15cm}
	\item Defending against arbitrary pirate attacks. \vspace{-0.15cm}
	\item The classical group testing model. \vspace{-0.1cm}
\end{enumerate}
These settings are studied in the following subsections.

\subsection{Defending against the interleaving attack}
\label{sec:int}

When defending against the interleaving attack, which may be the most practical pirate strategy due to its simplicity and its strength, the Neyman-Pearson decoder of~\eqref{eq:g} designed against the interleaving attack~\cite{laarhoven14ihmmsec} may be a good choice, even when the tracer's estimate $c_0$ is not exact~\cite{furon14}. With this score function, in each segment the average pirate score increases by $\mu_1 \sim \frac{1}{2 c^2}$, while for innocent users we have $\mu_0 \sim \frac{-1}{2 c^2}$ which is proved in the appendix. \\

In Wald's scheme, recall that we may set $\eta_1 = \ln(1/\eps_1')$ conservatively where $\eps_1'$ is the per-user false positive probability. Without using a lower threshold, setting $\eps_1' = \eps_1/n$, i.e., $\eta_1 = \ln(n/\eps_1)$ (and $\eta_0 = -\infty$) guarantees that with probability at least $1 - \eps_1$, no innocent users are ever accused, and with probability $1$ all colluders are eventually found. We illustrate the scheme with a toy example in Figure~\ref{fig:ex1wald}, where we set the parameters as $n = 1000$, $c_0 = c = 10$, and $\eps_1 = 10^{-3}$, so that $\eta_1 \approx 13.82$. On average, it takes about $3000$ segments to trace the colluders.

\begin{figure}[t]
\centering
\subfloat[][Example of Wald's scheme: Interleaving attack]{\includegraphics[width=8cm]{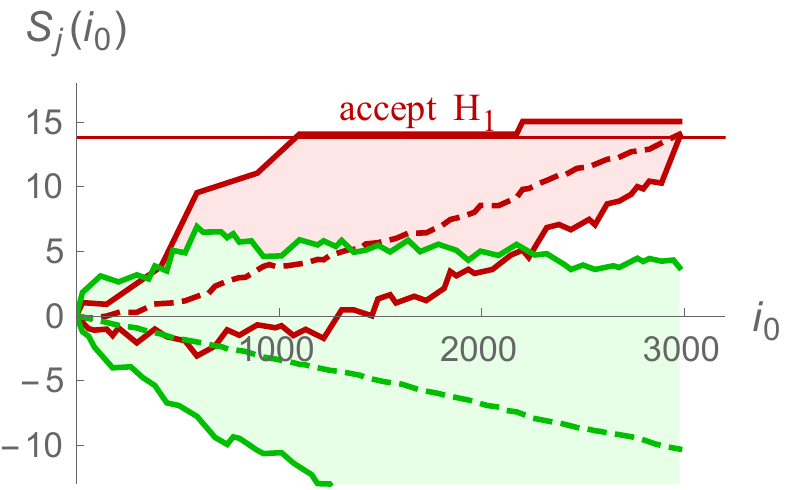}\label{fig:ex1wald}} \\
\subfloat[][Example of Tardos' scheme: Interleaving attack]{\includegraphics[width=8cm]{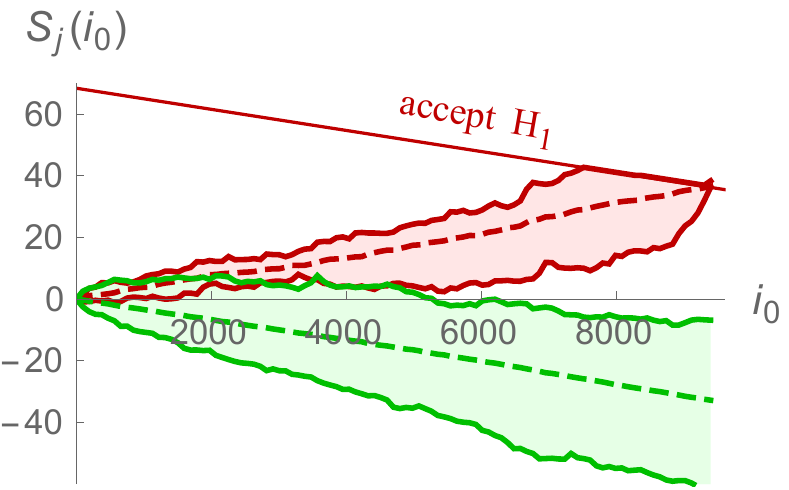}\label{fig:ex1tardos}}
\caption{A simulation of the interleaving attack in fingerprinting with Wald's SPRT (above) and the sequential Tardos scheme with log-scores (below). Even though it might be caused by not so tight parameter choices in the sequential Tardos scheme, with Wald's scheme we can more easily select tight parameters and trace colluders faster. \label{fig:ex1}}
\end{figure}

For large $n$ and $c$, this corresponds to an asymptotic upper threshold of $\eta_1 \sim \ln n$ and an expected time of $\ell \sim 2 c^2 \ln n$ until all pirates have been found. This thus corresponds to drawing a horizontal accusation threshold starting at $(0, \ln n)$, and the pirates are expected to be found around the point $(2c^2 \ln n, \ln n)$. \\

In the sequential Tardos scheme with log-likelihood scores, setting the parameters is more difficult. Various provable bounds on the error probabilities for given parameters are not tight, leading to pessimistic estimates and higher thresholds and code lengths than required. One could also estimate the actual required code length for a given set of parameters directly, leading to better scheme parameters, but this would have to be done for each instance separately; if any of the parameters $c, n, \eps_1, \eps_2$ then changes, one would have to redo the simulations or computations to find good practical parameters for the new setting. 

To illustrate how far the provable parameters are off from reality, we again used the toy example of Figure~\ref{fig:ex1wald} and used the provable bounds from \cite[Theorem 3]{laarhoven14ihmmsec} and Theorem~\ref{thm:adaptive1}, to obtain the following parameters:
\begin{align}
&\qquad \quad \mu_0 \approx 0.00382, \qquad \mu_1 \approx -0.00343, \\
&\ell \approx 17\,953, \qquad \eta_1^{(\ell)} \approx 6.9078, \qquad \eta_1^{(0)} \approx 68.41.
\end{align}
Note that $\eta_1^{(0)}$, the value of the boundary at $i_0 = 0$, is significantly higher than when using Wald's SPRT. An illustration of how this scheme would then work in practice is given in Figure~\ref{fig:ex1tardos}. In most cases the scheme finds all pirates after roughly $9000$ segments, which although much less than the provable code length of $\ell \approx 18\,000$ is much higher than the practical code length of the SPRT of $\ell \approx 3000$ segments.

For large $n$ and $c$ we can again compute what the parameters converge to. First, for this setting we also obtain an asymptotic code length of $\ell \sim 2 c^2 \ln n$, and this again corresponds to the asymptotic point $(2 c^2 \ln n, \ln n)$; see e.g.\ \cite[Theorem~3, Corollary~3]{laarhoven14ihmmsec}. However, in the sequential Tardos scheme this accusation threshold is a decreasing line (cf.~Figure~\ref{fig:sketch2}) with a slope equal to $\mu_0 \sim \frac{-1}{2c^2}$. This means that at time $i_0 = 0$, the accusation line asymptotically starts at $\ln n - (2 c^2 \ln n) \cdot \frac{1}{2c^2} = 2 \ln n$; the red line starts twice as high as in Wald's SPRT, at the point $(0, 2 \ln n)$. So although both schemes achieve the same asymptotic code length, even in the limit of large $n$ and $c$ these schemes are slightly different.

In this case there are several reasons to prefer Wald's SPRT approach: it is easier to choose good parameters, and asymptotically the accusation threshold lies lower than in the sequential Tardos scheme, allowing for a slightly faster tracing of the colluders. 

\subsection{Defending against arbitrary pirate attacks}

For the general, uninformed fingerprinting game, where it is not known what collusion channel was used, the tracer has to use a decoder that works well against arbitrary collusion channels. Again, the paper of Furon and Desoubeaux~\cite{furon14} compares various of these candidate decoders, each of which were derived through different optimization techniques. 

The decoder described in the previous subsection, designed against the interleaving attack, is capacity-achieving in the uninformed setting as well~\cite{laarhoven14ihmmsec}, and so a similar construction as in the previous subsection may again be used, both in the sequential Tardos scheme and in Wald's sequential scheme. As described in~\cite{laarhoven13wifs}, with this score function it can only be guaranteed that $(\mu_1 - \mu_0) / \sigma_0$ is sufficiently large regardless of the collusion channel, i.e., it is possible to distinguish between the innocent and guilty distributions. However, it could be that for different collusion channels, both $\mu_0$ and $\mu_1$ are smaller than when the interleaving attack is used. To cope with these difficulties, one could normalize the scores, i.e., based on $y_i$ and $p_i$, compute $\mu_0$ and $\sigma_0$ for segment $i$, and translate and scale the scores so that $\mu_0$ and $\sigma_0$ are the same as for the interleaving attack.

Alternatively, one could use a wide range of different methods, such as using several score functions simultaneously; estimating the collusion channel and using this estimate to choose the score function~\cite{charpentier09, furon09b}; using a generalized linear decoder~\cite{abbe12, desoubeaux13, meerwald12}; or settle for slightly less and use the suboptimal but `universal' symmetric score function of \v{S}kori\'{c} et al.\ \cite{skoric08} which works almost the same for any collusion strategy. For small collusion sizes this score function does not perform that poorly~\cite[Figure~3]{furon14}, and it might make designing the scheme somewhat easier.

\subsection{The classical group testing model}

Let us further highlight how the sequential Tardos scheme and Wald's SPRT can behave very differently, by showing how both schemes apply to the classical group testing model. In group testing~\cite{dorfman43} one is tasked to identify the defective members $\mathcal{C}$ from a population $\mathcal{U}$ by performing group tests: testing a query group $\mathcal{Q} \subseteq \mathcal{U}$ returns a positive result if $\mathcal{Q} \cap \mathcal{C} \neq \emptyset$ and a negative result otherwise. Applications include blood testing for viruses, where pooling blood samples of several persons and testing this batch leads to a positive test result iff the virus is present in the tested batch. In terms of fingerprinting, this problem corresponds to dealing with the all-$1$ attack~\cite{laarhoven13allerton, laarhoven14ihmmsec, meerwald11}.

As described in~\cite{laarhoven14ihmmsec}, the Neyman-Pearson approach to the all-$1$ attack in fingerprinting leads to the following optimized decoder $g$:
\begin{align}
g(x,y) = \begin{cases} 
\frac{1}{c} \ln(2) & \text{if } (x, y) = (0, 0); \\
\ln\left(2 - 2^{-1/c}\right) & \text{if } (x, y) = (0, 1); \\
-\infty & \text{if } (x, y) = (1, 0); \\
\ln(2) & \text{if } (x, y) = (1, 1). \end{cases} \label{eq:gt-scores}
\end{align}
Note that this function does not depend on $p$ anymore; to deal with the all-$1$ attack, it is best to replace the random bias generation using the arcsine distribution with a fixed bias $p \sim \frac{\ln 2}{c}$. In the non-adaptive, simple decoding setting, this leads to a required code length of $\ell \sim \frac{c \ln n}{(\ln 2)^2}$, while in the joint decoding setting the required code length becomes $\ell \sim c \log_2 n$, a factor $\ln 2$ less.

In Wald's scheme, choosing scheme parameters is done similarly as in Section~\ref{sec:int}. If we again consider the toy application of $c = 10$, $n = 1000$ and $\eps_1 = 10^{-3}$ (with $\eps_2 = 0$, not using a lower boundary), then we may again set $\eta_1 \approx 13.82$ and we are ready to use the scheme. Figure~\ref{fig:ex2wald} shows an example simulation of this scheme with these parameters, using the all-$1$ score function from~\eqref{eq:gt-scores}. For simplicity, we used the asymptotic approximation $p = (\ln 2)/c$ for generating the code.

\begin{figure}[t]
\centering
\subfloat[][Example of Wald's scheme: All-$1$ attack]{\includegraphics[width=8cm]{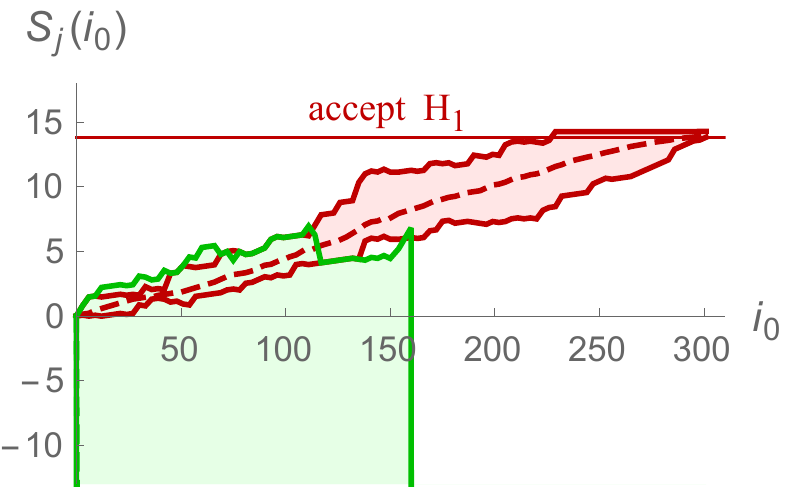}\label{fig:ex2wald}} \\
\subfloat[][Example of Tardos' scheme: All-$1$ attack]{\includegraphics[width=8cm]{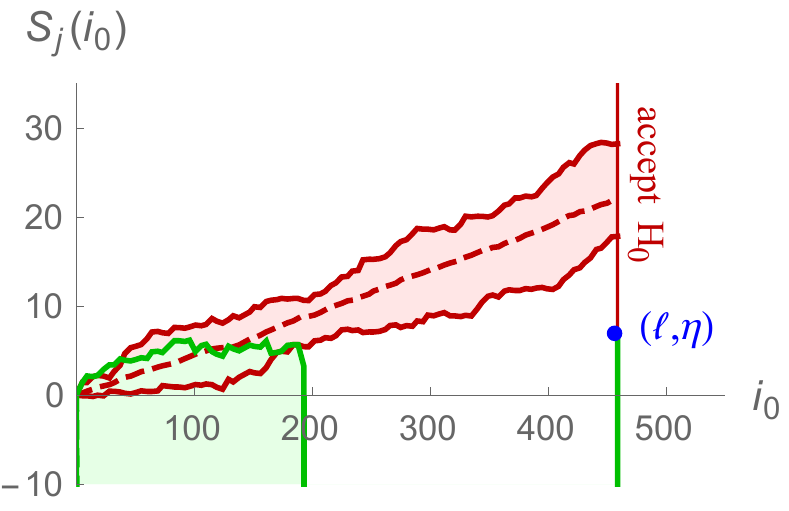}\label{fig:ex2tardos}}
\caption{A simulation of the classical group testing model. The sequential Tardos scheme offers no improvement over non-adaptive group testing, and again Wald's scheme finds the defectives faster. \label{fig:ex2}}
\end{figure}

In the sequential Tardos scheme, one would again first determine the point $(\ell, \eta)$ at which time a decision is taken(cf.\ Figure~\ref{fig:sketch2}), and then draw the accusation threshold by drawing a line towards the $y$-axis, parallel to the line with the average innocent user scores. Note however that with this score function, the event $(x,y) = (1,0)$ is impossible for a guilty user (if a member $j \in \mathcal{C}$ has $(\vc{x}_j)_i = 1$, then by definition $y_i = 1$ as well), and so if this event occurs we know for sure that this user is innocent, and we assign the user a score of $g(1,0) = -\infty$. Since the probability that this event occurs for innocent users is a positive constant, it immediately follows that $\mu_0 = -\infty$. As a consequence, the accusation threshold starting from $(\ell, \eta)$ with a slope of $-\mu_0 = \infty$ becomes a vertical line upwards. This is illustrated in Figure~\ref{fig:ex2tardos} which shows an example application of the sequential Tardos scheme with the same parameters as in Wald's scheme. The provable bounds from~\cite[Theorem 3]{laarhoven14ihmmsec} and Theorem~\ref{thm:adaptive1} lead to $\ell \approx 459$ and $\eta_1^{(\ell)} \approx 6.91$.


\section{Wald vs. Tardos: An overview}
\label{sec:overview}

Let us conclude with a brief overview of the two solutions for the sequential fingerprinting game. For simplicity, we will compare the sequential Tardos scheme with Wald's scheme without a lower boundary, as that seems to be the most convenient solution in fingerprinting. Note that although the sequential Tardos scheme is different from Wald's basic description of the sequential probability ratio test procedure, it could be considered a variant of the latter; truncating the thresholds was already considered by Wald himself, and using different shapes for the stopping boundary has been discussed extensively in various literature on the SPRT. 

\begin{table}[!t]
\renewcommand{\arraystretch}{1.3}
\centering
\begin{tabular}{p{5cm}cc} \toprule
\textbf{Characteristic} & \textbf{Wald} & \textbf{Tardos} \\ \midrule
Optimal (cf.\ Theorem~\ref{thm:adaptive2}) & \cmark & \xmark \\
Asymptotically optimal & \cmark & \cmark \\
No false negatives ($\eps_2 = 0$) & \cmark & \xmark \\
Guaranteed decision at time $\ell$ & \xmark & \cmark \\
Parameters to choose & $\eta_1$ & $\ell, \eta_1^{(0)}, \eta_1^{(\ell)}$ \\
Simple relations for parameters & \cmark & \xmark \\
Better than non-ad. gr. testing & \cmark & \xmark \\
\bottomrule
\end{tabular}
\caption{A quick summary of various characteristics of Wald's SPRT and the sequential Tardos scheme. For Wald's scheme we assume we are not using a lower boundary $\eta_0$, i.e., we set $\eps_2 = 0$ and $\eta_0 = -\infty$.\label{tab:tab1}}
\end{table}

To compare some of the characteristics, Table~\ref{tab:tab1} gives a quick summary of the various characteristics of both schemes. Here optimality refers to the optimality described in Theorem~\ref{thm:adaptive2}, and asymptotic optimality refers to the large $n$ and large $c \ll n$ regime. Note that by setting $\eps_2 = 0$ in Wald's scheme, we guarantee that eventually all colluders are always caught. This solution of an infinite accusation boundary comes at the cost of not knowing in advance how many segments are at most needed to reach a decision, although in practice this does not seem to be an issue. As we saw in Section~\ref{sec:comp}, often choosing parameters is easier for Wald's scheme than for the sequential Tardos scheme; both because fewer parameters have to be chosen, and because there is a simple approximate relation between this single parameter $\eta_1$ and the error probability $\eps_1$, which holds exactly if the scores behave like true Brownian motions. We further saw that for the sequential group testing setting, the sequential Tardos scheme offers no improvement over non-adaptive decoding (while Wald's scheme does).

Finally, as mentioned before, Wald's scheme has already been studied since the 1940s, with many books appearing on the topic ever since~\cite{bartroff13, chernoff72, govindarajulu04, jennison00, mukhopadhyay09, siegmund85, wald47, wetherill86}, while the sequential Tardos scheme~\cite{laarhoven13tit} was more of an ad-hoc construction to build a scheme that works well in adaptive settings as well. With Wald's scheme being easier to design, in many cases performing better than the sequential Tardos scheme (performing optimally well), and being backed by decades of research on the topic (allowing practitioners to tweak the scheme to their needs using known results from the literature), it seems that Wald's scheme is a more practical choice than the sequential Tardos scheme. 

In this paper we further settled an important question on the optimality of these schemes (i.e., both schemes are asymptotically optimal in the sequential setting), but one important open question remains: is it possible to design truly adaptive fingerprinting schemes that perform even better than these sequential designs? Or are the sequential schemes discussed in this paper also optimal in the adaptive setting? This is left as an open problem for future work.


\section*{Acknowledgments}

The author thanks Boris \v{S}kori\'{c} for various useful comments on an initial version of this manuscript.


\appendix

\section{Expected innocent scores for the interleaving attack}
\label{app:mu1}

\begin{lemma}
Suppose that:
\begin{itemize}
  \item the encoder uses the arcsine distribution encoder;
  \item the collusion channel is the interleaving attack $\thint$;
  \item the decoder is the interleaving log-likelihood decoder. 
\end{itemize}
Then the expected score of an innocent user ($\mu_1$) in a single segment is asymptotically given by:
\begin{align}
\mu_0 &= \expn_{x,y,p} \big[g(x,y,p) \mid H_1\big] 
\sim -\frac{1}{2c^2}.
\end{align}
\end{lemma}

\begin{proof}
For $\mu_0$, we write out the expectation:
\begin{align}
\mu_0 &= \expn_{x,y,p} \big[g(x,y,p) \mid H_1\big] \\
&= \int_0^1 \frac{\textrm{d}p}{\pi \sqrt{p(1-p)}} \left[p^2 \ln\left(1 + \frac{1-p}{cp}\right) \right.\\
&+ \left.2p(1-p) \ln\left(1 - \frac{1}{c}\right) + (1-p)^2 \ln\left(1 + \frac{p}{c(1-p)}\right)\right].
\end{align}
Similarly, we can write out the definition of the average colluder score in a single segment as:
\begin{align}
\mu_1 &= \int_0^1 \frac{\textrm{d}p}{\pi \sqrt{p(1-p)}} \left[p^2 \left(1 + \frac{1-p}{cp}\right) \ln\left(1 + \frac{1-p}{cp}\right) \right. \\
&+ 2p(1-p) \left(1 - \frac{1}{c}\right) \ln\left(1 - \frac{1}{c}\right) \\
&+ \left.(1-p)^2 \left(1 + \frac{p}{c(1-p)}\right) \ln\left(1 + \frac{p}{c(1-p)}\right)\right].
\end{align}
Combining these results, and merging the logarithms into one term, we obtain the following expression for $\mu_1 - \mu_0$:
\begin{align}
\mu_1 - \mu_0 
 &= \int_0^1 \frac{\sqrt{p(1-p)} \, \textrm{d}p}{\pi c} \ln\left(1 + \frac{c}{(c-1)^2 p(1-p)}\right).
\end{align}
Since we know that $\mu_1 \sim \frac{1}{2c^2}$, we need to prove that the right hand side is asymptotically similar to $\frac{1}{c^2}$. Rearranging terms, we thus need to prove that
\begin{align}
I \stackrel{\operatorname{def}}{=} \int_0^1 \textrm{d}p \sqrt{p(1-p)} \ln\left(1 + \frac{c}{(c-1)^2 p(1-p)}\right) \sim \frac{\pi}{c}.
\end{align}
First, using $\ln(1 + x) < x$ for all $x > 0$, we obtain:
\begin{align}
I < \int_0^1 \frac{c \, \textrm{d}p}{(c-1)^2 \sqrt{p(1-p)}} = \frac{\pi c}{(c-1)^2} \sim \frac{\pi}{c}.
\end{align}
To get a matching lower bound, we first reduce the range of integration from $[0,1]$ to $[\delta, 1 - \delta]$ for some $\delta > 0$, noting that the integrand is strictly positive:
\begin{align}
I > \int_{\delta}^{1 - \delta} \textrm{d}p \sqrt{p(1-p)} \ln\left(1 + \frac{c}{(c-1)^2 p(1-p)}\right).
\end{align}
Choosing $\delta = \frac{1}{\sqrt{c}}$, the term inside the logarithm is small and the following bound is tight enough to obtain the result:
\begin{align}
I \gtrsim \int_{\delta}^{1 - \delta} \frac{c \, \textrm{d}p}{(c-1)^2 \sqrt{p(1-p)}} \sim \frac{\pi}{c} - \frac{4}{\pi c} \arcsin \sqrt{\delta} \to \frac{\pi}{c}.
\end{align}
This proves that $I \sim \frac{\pi}{c}$, hence $\mu_0 \sim -\frac{1}{2c^2}$.
\end{proof}

\end{document}